\documentclass[a4paper, UKenglish,cleveref, autoref, thm-restate]{lipics-v2021}

\pdfoutput=1 
\hideLIPIcs  

\title{Edit Distance of Finite-Valued Transducers} 

\author{Prince Mathew}{Universit\'e libre de Bruxelles, Belgium}{prince.mathew@ulb.be}{0000-0001-6410-1474}{}
\author{Saina Sunny}{Universit\'e libre de Bruxelles, Belgium}{saina.sunny@ulb.be}{0009-0005-1366-0168}{}

\authorrunning{P. Mathew, and S. Sunny} 

\Copyright{Prince Mathew and Saina Sunny} 




\acknowledgements{We would like to thank all the anonymous reviewers for their in-depth feedback which contributed to improving this paper.}

\nolinenumbers 

\usepackage{xspace}
\usepackage{graphicx} 
\usepackage{amsthm}
\usepackage{amsmath}
\usepackage{amsfonts}
\usepackage{mathtools}
\usepackage{xcolor}
\usepackage{url}
\usepackage{todonotes}
\usepackage{etoolbox}

\usepackage{tikz}
\usetikzlibrary{automata, positioning, arrows,calc}

\newcommand\dom{\mathit{dom}}

\newcommand\calA{\mathcal{A}}
\newcommand\calD{\mathcal{D}}

\newcommand\calS{\mathcal{S}}
\newcommand\calT{\mathcal{T}}
\newcommand\calP{\mathcal{P}}

\newcommand{\omicron}{\mathit{o}}
\newcommand{\N}{\mathbb{N}}
\newcommand{\autom}{\ensuremath\cal{A}}
\newcommand\CF[1]{\ensuremath\mathsf{#1}}
\renewcommand{\omicron}{\scalebox{.8}{$\mathcal{O}$}}

\renewcommand{\arraystretch}{1.5} 

\newcommand{\dlev}{\ensuremath{d_l}}

\EventEditors{Sayan Bhattacharya, Danupon Nanongkai, Michael Benedikt, and Gabriele Puppis}
\EventNoEds{4}
\EventLongTitle{53rd International Colloquium on Automata, Languages, and Programming (ICALP 2026)}
\EventShortTitle{ICALP 2026}
\EventAcronym{ICALP}
\EventYear{2026}
\EventDate{July 7--10, 2026}
\EventLocation{Royal Holloway, University of London, Egham, United Kingdom}
\EventLogo{}
\SeriesVolume{374}
\ArticleNo{189}

\begin{document}

\maketitle

\begin{abstract}
Transducers generalise automata by producing output word(s) for each input word, thereby defining a relation over words. 
A transducer is said to be \emph{finite-valued} if, for every input word, it produces at most \(k\) output words, for some constant \(k\). If \(k = 1\), then the transducer is said to be \emph{functional}. The \emph{edit distance} between two transducers is the minimal number of edits required to transform every output of one transducer into some output of the other, for each input word. This notion has been studied for functional transducers, where it is shown to be computable. However, it is uncomputable for transducers in general.  In this work, we show the computability of the edit distance of finite-valued transducers, a class that is strictly more expressive than functional transducers.
  
\end{abstract}

\section{Introduction}
The theory of automata and transducers provides a rich foundation for modelling and reasoning about languages and relations and has numerous applications, particularly in program synthesis and the formal verification of hardware and software systems. While automata map words to Boolean values, word transducers extend this framework by modelling transformations: they read input words and produce output words using finite memory, thereby defining relations between words, known as rational relations.
Their study dates back to the early days of computer science, where they were referred to as generalised sequential machines~\cite{Raney1958SequentialF,Ginsburg1968-GINACO-2, DBLP:conf/stacs/MuschollP19,DBLP:journals/siglog/FiliotR16} {and the references therein}. Their ability to capture input–output behaviour makes them a natural tool in diverse applications such as natural language processing, formal verification, and program analysis.

A notable subclass of transducers is that of finite-valued transducers. They compute rational relations where each input word is mapped to at most $k$ output words, for a fixed $k\in\N$. Such relations are called $k$-valued (or \emph{finite-valued}) rational relations. When $k=1$, the rational relation is called a \emph{rational function}. Many problems that are undecidable in general for transducers become decidable under finite-valuedness. For instance, inclusion and equivalence are undecidable in general~\cite{FischerR68}, but become decidable under the finite-valued restriction \cite{CulikEquivalence,Weber93,Souza08}.
The subclass of finite-valued relations is decidable within the class of rational relations, i.e., given a finite state transducer $\calT$, it is decidable in polynomial time whether $\calT$ is a finite-valued transducer~\cite{Weber89,SakarovitchS08MFCS}. In fact, it is also decidable in polynomial time to check if $\calT$ recognises a $k$-valued rational relation for a given $k\in\N$~\cite{GurariI83,SakarovitchS08MFCS}. Thus, finite-valued transducers strike a balance between expressive power and algorithmic tractability.

Recently, a quantitative notion of distance between transducers was introduced in \cite{editdistance} by lifting standard word metrics, most notably the edit distance, from words to word-to-word transformations. The edit distance measures the minimum number of edits (such as inserting or deleting a letter, or substituting a letter with another) required to transform one word into another. The edit distance between transducers provides such a quantitative measure: for every input word, it asks how many edits are required to transform every output of one transducer into some output of the other. This notion generalises the classical Boolean equivalence problem of two transducers to a numerical setting, where two transducers are assigned a distance, and are considered close if their distance is finite (reducing to standard equivalence when the distance is zero). This arises naturally in applications such as comparing text-based tools, such as spell checkers or grammatical correction tools. This distance between transducers was used in the study of approximate counterparts of fundamental decision problems for rational relations—such as functionality, determinisability, and uniformisation~\cite{approximate}. It is shown in~\cite{editdistance} that the edit distance of functional transducers is computable; however, the problem is uncomputable for transducers in general~\cite{Sainathesis}.

This naturally raises the question: Is there an expressive yet well-behaved class of transducers, beyond the functional case, where edit distance remains computable? In this work, we answer this question affirmatively for finite-valued transducers, showing that their edit distance is computable. 
This result extends the computability of edit distance beyond functional transducers and provides a foundation for approximate reasoning within a class of transducers that is both expressive and practically motivated. 

The procedure for computing the edit distance between finite-valued transducers exploits the fact that every finite-valued transducer can be decomposed into a finite union of functional transducers \cite{Weber93,SakarovitchS08}. 
We reduce the problem of computing the distance between finite-valued transducers to that of computing the distance between multi-sequential transducers (equivalent to a finite union of sequential transducers --- those that are deterministic in the input). This further reduces to a new notion we introduce in this work, called the \emph{relative distance} between a function and a relation.
The \emph{relative distance} between a function \( f \) and a relation \( R \) is defined as the least upper bound, over all inputs, of the minimal edit distance between the output of \( f \) and \emph{some} output of \( R \) on the same input. Intuitively, this captures how far, in the worst case, the behaviour of \( f \) deviates from that of \( R \). 
Beyond its technical role in our reductions, relative distance also provides a natural framework for \emph{approximate verification}. In classical (exact) verification of reactive systems (systems that continuously interact with an environment by consuming inputs and producing outputs), it is checked whether the given reactive system (modelled by $f$) exactly satisfies a user-defined specification (modelled by $R$). The notion of relative distance makes it possible to reason about reactive systems that are ``close'' to a given specification, which is crucial when exact conformance fails, and minor deviations to the specification are acceptable.

In summary, we extend the known results on functional transducers to finite-valued transducers, showing that their edit distance is computable. The techniques we develop not only advance the theory of transducer comparison but also provide tools for approximate reasoning and verification of transducer-defined relations.  

\section{Preliminaries}
\label{sec:prelims}
We now recall the basic notions and notations that will be used throughout the paper. Let $\mathbb{N}$ denote the set of natural numbers. 
For every $k \in \mathbb{N}$, we use $[k]$ to denote the set $\{1,2,\ldots, k\}$.

\subparagraph*{Words and relations.} Let $A$ and $B$ denote finite alphabets of letters. A \emph{word} is a sequence of letters. The empty word is denoted by $\varepsilon$. The length of a word is denoted by $|w|$.
The set of all finite words over the alphabet $A$ is denoted by $A^*$.  
A binary \emph{relation} $R$ between two sets $U$ and $V$ is a subset of the Cartesian product $U \times V$. The domain of $R$ is defined as $\dom(R) = \{u \mid (u,v) \in R\}$. For $u \in U$, the set of elements related to $u$ under $R$ is denoted as $R(u) = \{v \mid (u,v) \in R\}$. Note that $u\in\dom(R)$ iff $R(u)$ is nonempty.  A relation $R$ is said to be \emph{$k$-valued} (resp. \emph{exactly $k$-valued}) for some $k \geq 1$, if for all $u \in \dom(R)$, the cardinality of the set $R(u)$ is at most $k$ (resp. exactly $k$). A \emph{function} $f: U \rightarrow V$ is a relation over $U \times V$ which is exactly 1-valued. The function $f$ is \emph{partial} if $\dom(f) \subseteq U$.

\subparagraph*{Finite state automata.} A finite state automaton is a fundamental model of computation that accepts a set of words using finite memory.
\begin{definition}
\label{def:fsa}
A (nondeterministic) finite state automaton $\autom$ is a tuple $(Q, A, I,\Delta,F), \text{ where }$ $Q$ is a finite set of states, 
 $A$ is the input alphabet,
     $I \subseteq Q$ is the set of initial states, 
 $\Delta \subseteq Q \times (A \cup \{\varepsilon\})  \times Q$ is the finite set of transitions, and 
 $F \subseteq Q$ is the set of final states.
\end{definition}
For $a \in A \cup \{\varepsilon\}$, we write $p \xrightarrow{a}q$ whenever $(p,a,q) \in \Delta$. A \emph{run} of $\autom$ on an {input} word $w=a_1 \cdots a_n$, where $a_i \in A \cup \{\varepsilon\}$ for all $i \in [n]$, is a sequence of transitions $q_0 \xrightarrow{a_1} q_1 \cdots \xrightarrow{a_n} q_n$. We use $q_0 \xrightarrow{w}q_n$ to denote this. The run is  \emph{accepting} if  $q_0 \in I$ and $q_n \in F$. The \emph{language} of $\autom$, denoted by $L(\autom)$, is the set of words $w$ s.t.~$\autom$ has an accepting run on $w$. 

We say that a finite state automaton is \emph{complete} if for every state $q \in Q$ and for every input symbol $a \in A$, there exists $q'\in Q$ s.t.~$(q, a, q') \in \Delta$. An automaton is \emph{real-time} if it has \emph{no} transitions of the form $p \xrightarrow{\varepsilon} q$.
An automaton is said to be \emph{trim} if, for any state $q$, there exists at least one accepting run visiting $q$. A finite state automaton is \emph{deterministic} if $I$ is a singleton set and the set of transitions is a (partial) function from $Q \times A$ to $Q$. The problem of checking equivalence is $\CF{PSPACE}$-complete for nondeterministic finite state automata (NFA)~\cite{stock73} and in polynomial time for deterministic~\cite{SteHun85} finite state automata (DFA).

\subparagraph*{Finite state transducers.} Transducers extend finite state automata by labelling each transition and each final state by an output word (possibly empty), thereby realising word-to-word relations.
\begin{definition} A (nondeterministic) finite state transducer $\calT$ with input alphabet $A$ and output alphabet $B$ is a tuple 
$(Q,A, B, I,\Delta, F, \lambda, \omicron) \text{ where,}$
 $(Q, A, I,\Delta,F)$  is a nondeterministic finite state automaton over $A$ (also called the \emph{underlying automaton} of $\calT$),
 $\lambda: \Delta \rightarrow B^*$ is a labelling function that maps transitions to output words over $B$, and 
 $\omicron: Q \rightarrow B^*$ is a labelling function that maps states to output words over $B$. 
\end{definition}
We write $p \xrightarrow{a \mid v}_{\calT}q$, if $(p,a,q) \in \Delta$ and $v = \lambda(p,a,q)$. When the transducer $\calT$ is clear from the context, we omit the subscript and simply write $p \xrightarrow{a \mid v} q$.  
A \emph{run} of $\calT$ labelled by an \emph{input} word $a_1 \cdots a_n$, where $a_i \in A \cup \{\varepsilon\}$ for all $i \in [n]$, is a sequence of transitions $q_0 \xrightarrow{a_1 \mid v_1} q_1 \cdots \xrightarrow{a_n \mid v_n} q_n$. We denote such a run by $q_0 \xrightarrow{a_1 \cdots a_n \mid v_1 \cdots v_n} q_n$. The \emph{output} word along this run is $v_1 \cdots v_n \cdot \omicron(q_n)$. 
The run is \emph{accepting} if $q_0 \in I$ and $q_n \in F$. A \emph{loop rooted at}  state $q$ of $\calT$ is a run of $\calT$ s.t.~the run starts and ends at $q$. The relation \emph{defined}/\emph{recognised} by $\calT$, denoted as $[\![ \calT ]\!]  \subseteq A^*\times B^*$, is defined as 
\begin{equation*}
[\![ \calT ]\!] = \{ (u,v)\mid v \in B^* \text{ is the output along some accepting run of $\calT$ labelled by $u$}\}.
\end{equation*}
Relations defined by transducers are called \emph{rational relations}. For convenience, we use $\calT(u)$ for $[\![\calT]\!](u)$, and use $\dom(\calT)$ to denote $\dom([\![\calT]\!])$, called the \emph{domain} of the transducer $\calT$. 
We say that a transducer is \emph{complete} (resp. \emph{real-time, trim}) if its underlying automaton is complete (resp. real-time, trim). The Cartesian product of two transducers $\calT_i = (Q_i, A, B, I_i, \Delta_i, F_i, \lambda_i, \omicron_i)$ for $i \in [2]$, denoted by $\calT_1 \times \calT_2$, is defined as the transducer $(Q_1 \times Q_2,\, A,\, (B \cup \{\varepsilon\}) \times (B \cup \{\varepsilon\}),\, I_1 \times I_2,\, \Delta,\, F_1 \times F_2,\, \lambda,\, \omicron) \text{ where }$ $((p_1,p_2),a,(q_1,q_2)) \in \Delta$ if $(p_i,a,q_i) \in \Delta_i$ for each $i \in [2]$, $\lambda((p_1,p_2),a,(q_1,q_2)) = \big(\lambda_1(p_1,a,q_1),\, \lambda_2(p_2,a,q_2)\big)$, and $\omicron(p_1,p_2) = \big(\omicron_1(p_1),\, \omicron_2(p_2)\big)$ for all $(p_1,p_2) \in Q_1 \times Q_2$.

We will now review some important subclasses of rational relations that capture natural restrictions on transducers encountered in this work.

\subparagraph*{Subclasses of rational relations.} Let $\calT$ be a real-time transducer. If $[\![ \calT ]\!]$ is $k$-valued for some $k \in \N$, then $\calT$ is said to be \emph{finite-valued} (or $k$-valued). The relations recognised by such transducers are called \emph{finite-valued rational relations}. For example, the $k$-subword relation $\{(u,v) \mid u,v \in \{a,b\}^*,\ v \text{ is a subword of } u \text{ of length } k\}$ is finite-valued. Since for any input word $u$, the number of possible outputs $v$ is bounded by $2^k$, corresponding to all words of length $k$ over the alphabet $\{a,b\}$.

If $[\![ \calT ]\!]$ is functional, then $\calT$ is said to be a \emph{functional transducer}. The functions recognised by functional transducers are called \emph{rational functions}. An example is the function $f_{last}: u \sigma \mapsto \sigma u$ where $u \in \{a,b\}^*, \sigma \in \{a,b\}$; that moves the last letter of the input word to the beginning. If the underlying automaton of $\calT$ is unambiguous (i.e., has at most one accepting run on any input), then $\calT$ is referred to as an \emph{unambiguous transducer}. It is well-known that a function is recognised by a real-time transducer iff it is recognised by a rational transducer~\cite{Berstel} iff it is recognised by an unambiguous transducer~\cite{eilenberg1974automata}.

A transducer is called \emph{sequential} if its underlying automaton is deterministic.  
Such transducers define functions known as \emph{sequential functions}. 
In this case, 
the transition relation $\Delta$ becomes a function $\Delta: Q \times A \to Q$. 
Sequential functions are a strict subset of rational functions. For instance, the function $f_{last}$ is not sequential since, intuitively,  the transducer recognising $f_{last}$ must guess the last letter initially.

Any finite-valued transducer is known to be (effectively) equivalent to a finite union of unambiguous transducers~\cite{Weber93}. Consequently, any finite-valued rational relation can be expressed as a finite union of rational functions.
A notable strict subclass of finite-valued rational relations is the class of \emph{multi-sequential relations}, which corresponds to a finite union of sequential functions~\cite{DBLP:journals/ijfcs/JeckerF18}. A transducer is said to be \emph{multi-sequential} if it is a finite union of state-disjoint sequential transducers. The function $f_{last}$ is multi-sequential as it can be expressed as the union of two sequential functions $f_a: ua \mapsto au$ and $f_b: ub \mapsto bu$. However, the function $f_{\sf last}^\ast$ that maps $u_1\# \cdots u_n\#$, for $n \geq 1$, to $f_{\sf last}(u_1)\# \cdots f_{\sf last}(u_n)\#$ for some separator $\#$ is $1$-valued, but not multi-sequential.

\subparagraph{Distances between words and word relations.} A metric $d$ on a set $E$ is a mapping
$d : E^2\rightarrow\mathbb{R}^+\cup \{\infty\}$ satisfying the following properties. For any $u,v,w \in E$, 
$d(u,v)=0 \iff u=v$ (separation),
 $d(u,v) = d(v,u)$ (symmetry), and
  $d(u,v) \leq d(u,w) + d(w,v)$ (triangle inequality). 
The most commonly studied metrics between finite words are the \emph{edit distances}. 
An edit distance between two words is the minimum number of edit operations required to rewrite one word into another if possible, and $\infty$ otherwise. Different choices of permitted operations give rise to different edit distances. \Cref{table:editdistance} lists some commonly used edit distances along with their corresponding sets of allowed operations. They are referred to as the \emph{Levenshtein family of distances}, as they are equivalent up to a constant factor \cite{editdistance}. 

{
\renewcommand{\arraystretch}{0.95}
\setlength{\tabcolsep}{5pt}
\begin{table}[h!]
\centering 
\begin{tabular}{|p{4.8cm}|p{1.2cm}|p{6cm}|}
\hline
Edit distances & Notation & Permissible edit operations \\
\hline
Longest Common Subsequence   & \ $d_{lcs}$   & insertions and deletions   \\
Levenshtein   & \ $\dlev$   & insertions, deletions and substitutions   \\
Damerau-Levenshtein   & \ $d_{dl}$   & insertions, deletions, substitutions and swapping adjacent letters   \\
\hline
\end{tabular}
\caption{Levenshtein family of edit distances and their allowed operations.}
\label{table:editdistance}
\end{table}
}

In the literature, the concept of distance between two words is naturally lifted to distance between a word and a set of words, or between two sets of words, and so on. The study of such extensions has received considerable attention in the literature.
The distance between two languages $L$ and $L'$ can be defined in a standard way, by resorting to the Hausdorff distance, denoted by $H_d(L,L')$. Informally, it is defined to be the least upper bound of all the distances from an element in one set to the ``closest element'' in the other set. 
In order to precisely define the Hausdorff distance, we first introduce the notion of directed distance. 
The \emph{directed distance} $\overset{\rightsquigarrow}{d}$ from $L$ to $L'$ is defined as
$\overset{\rightsquigarrow}{d}(L, L') =  \sup_{w \in L} \inf_{w' \in L'} d(w,w')$.
The \emph{Hausdorff distance} between $L$ and $L'$ is defined as
$H_d(L,L') = \max \left \{\overset{\rightsquigarrow}{d}(L, L'), \overset{\rightsquigarrow}{d}(L', L) \right \}$.  
Distances between words and languages can be lifted to that between relations over words. 

\begin{definition}[Distance between relations]\label{def:distrel}
 Given a metric $d$ on words, and two relations $R, S\subseteq A^*\times B^*$, the distance between $R$ and $S$, denoted as $d(R, S)$, is defined as follows.
 \[
d(R,S) =
\begin{cases}
\displaystyle
\sup \Bigl\{ H_d\bigl(R(w), S(w)\bigr) \;\Big|\; w \in \dom(R) \Bigr\}
& \text{if } \dom(R)=\dom(S),\\[6pt]
\infty
& \text{otherwise}.
\end{cases}
\]
\end{definition}
It was shown in~\cite{approximate, Sainathesis} that $d$ defines a metric over relations.  
Consequently, $R$ and $S$ are equivalent iff $d(R,S)= 0$. Observe that $d(R,S)<\infty$ iff $\dom(R) = \dom(S)$ and there exists a constant  $k\geq 0$ s.t.~for every word $u$ in the domain and for every output $v_1 \in R(u)$, there exists some output $v_2 \in S(u)$ with $d(v_1,v_2)\leq k$, and vice-versa. 

\begin{example}
\label{app:exdistance}
Consider the following rational relations $R$ and $S$ over $\{a,b\}^* \times \{a,b\}^*$ and $d$ to be a metric given in \Cref{table:editdistance}.
\begin{enumerate}

\item Let 
$R: w \mapsto \{a^{|w|},\, b^{|w|}\}, \ 
S: w \mapsto \{a^{|w|}\}.
$
For any input $w$, we have $R(w) = \{a^{|w|}, b^{|w|}\}$ and $S(w) = \{a^{|w|}\}$.  
Clearly, $d(a^{|w|},a^{|w|}) = 0$ and $d(b^{|w|},a^{|w|}) = |w|$. Hence,
\[
H_d\bigl(R(w),S(w)\bigr) \;=\; 
\max \left\{ \overset{\rightsquigarrow}{d}\bigl(R(w), S(w)\bigr), \ \overset{\rightsquigarrow}{d}\bigl(S(w), R(w)\bigr) \right\} 
= \max\{|w|,0\} = |w|.
\] 
As the input length $|w|$ increases, $H_d\bigl(R(w),S(w)\bigr)$ increases unboundedly (since the edit distance between $b^{|w|}$ and $a^{|w|}$ increases with $|w|$). Therefore, $d(R,S) = \infty$. 

\item Let 
$
R: w \mapsto \{(ab)^{|w|},\, (ba)^{|w|}\}, 
\  
S: w \mapsto \{(ab)^{|w|}\}.
$
For any input $w$, we have $R(w) = \{(ab)^{|w|}, (ba)^{|w|}\}$ and $S(w) = \{(ab)^{|w|}\}$. 
One can observe that $d\big((ab)^{|w|},(ab)^{|w|}\big)=0$, while 
$d\big((ba)^{|w|},(ab)^{|w|}\big)=2$ — obtained by deleting the initial $b$ and inserting a $b$ at the end — whenever $|w|>0$, and is $0$ otherwise. 
Thus, $H_d\bigl(R(w),S(w)\bigr) \leq  2$ for any input $w$, and hence $d(R,S) = 2$.
\end{enumerate}
\end{example}

The \emph{distance between transducers} $\calT$ and $\calS$, denoted by $d(\calT,\calS)$, is defined as the distance between the relations they recognise. 
In the case of edit distances, finite distance means that $\dom(\calT)= \dom(\calS)$ and for all input $u \in \dom(\calT)$, any output of $\calT$ on $u$ can be converted to an output of $\calS$ on $u$ by doing a bounded number of edits, and conversely, any output of $\calS$ on $u$ 
can be transformed into some output of $\calT$ on $u$ with a bounded number of edits.  This is relevant for comparing transducers with text-based output, such as spell checkers.

In general, the distance between two transducers is uncomputable. This follows from the undecidability of the equivalence problem for transducers~\cite{FischerR68}, which reduces to checking whether the distance is zero. Notably, it has been shown in~\cite{editdistance} that the distance between two functional transducers (or rational functions) is computable with respect to various edit distances, including the Levenshtein family of distances (See \Cref{table:editdistance}).
An analysis of their proof gives a doubly exponential space upper bound for the Levenshtein family of distances. The computability of distance in this context relies on the notion of \emph{conjugacy} of words~\cite{decidingconjugacy}. 

\begin{definition}[Conjugacy]
    Two words $u$ and $v$ are \emph{conjugate}, denoted by $u \sim  v$, if there exist words $x$ and $y$ s.t.~$u=xy$ and $v =yx$. In other words, they are cyclic shifts of each other.
\end{definition} 
For example, the words $aabb$ and $bbaa$ are conjugate with $x = aa$ and $y = bb$, whereas $aabb$ and $abab$ are not conjugate. The conjugacy defines an equivalence relation over words.
\begin{proposition}\label{prop:conj}
    Let $x,y,x',y',u,v$ be words and $C \in \mathbb{N}$. For any metric $d \in \{\dlev, d_{lcs}, d_{dl}\}$, \\ $d(xu^ky,x'v^{k}y') \leq C$ for all $k$ in some infinite subset $\mathbb{I} \subseteq \mathbb{N}$, iff $u$ and $v$ are conjugate.
\end{proposition}
\begin{proof}
   The proof for ($\Rightarrow$) is an adaptation of the proof of  Proposition~2 from \cite{approximate}. Since $d(xu^ky,x'v^{k}y') \leq C$, we get $|u| = |v|$. Otherwise, as $k \in \mathbb{I}$ increases, the length difference of $xu^{k}y$ and $x'v^{k}y'$ increases, and hence their distance will not be bounded. 
    Since $|u| = |v|$, either both $u$ and $v$ are nonempty, or $u=v= \epsilon$. In the latter case, $u$ and $v$ are conjugate. Assume that $u$ and $v$ are nonempty words. Take $k \in \mathbb{I}$ s.t.~$k \geq 2^{|u| + |v|}$. Since $d(xu^ky,x'v^{k}y') \leq C$, there exist large portions of $u$’s and $v$’s that match. In fact, $u$'s and $v$'s overlap at least of length $|u| + |v|$. By Fine and Wilf's\footnote{The Fine and Wilf's theorem states that if some powers of two words $u$ and $v$ share a common factor of length $|u| + |v| - \gcd(u,v)$ then their primitive roots are conjugate.} theorem, the primitive roots\footnote{The primitive root of a word $w$, denoted by $\rho_w$, is the shortest word s.t.~$w=(\rho_w)^n$ for some $n\in\N$.} of $u$ and $v$ are conjugate. Since $|u|=|v|$, it follows that $u$ and $v$ themselves are conjugate.

    For the other direction ($\Leftarrow$), assume $u$ and $v$ are conjugate, i.e., there exist words $p,q$ s.t.~$u=pq$ and $v=qp$. Hence, for every $k\geq 1$,  
    $d(xu^ky,x'v^ky') = d(x(pq)^ky,x'(qp)^ky') \leq d(xpy,x'py')$.
    This is finite for $d \in \{\dlev, d_{lcs}, d_{dl}\}$. 
\end{proof}

\begin{proposition}[\cite{editdistance}]\label{prop:levboundconj}
  Let $\calT$ and $\calS$ be functional transducers. For $d \in \{\dlev, d_{lcs}, d_{dl}\}$, $d(\calT,\calS) < \infty$ iff $\dom(\calT) = \dom(\calS)$ and every pair of output words generated by loops in the (trim) Cartesian product of $\calT$ and $\calS$ is conjugate. 
\end{proposition}

\section{Computing distance of finite-valued transducers}
\label{sec:distFiniteValued}
In this section, we address the problem of computing the distance between two finite-valued transducers. 
Our approach proceeds in two main steps. First, we show how to convert finite-valued transducers 
into multi-sequential transducers in such a way that the distance between the original transducers 
is preserved. Second, we show that the distance between multi-sequential transducers can be computed 
via a notion of relative distance, whose computability will be proved in the next section.
\subsection{Edit distance: finite-valued to multi-sequential}
We begin by reducing the problem of computing the edit distance between two finite-valued 
transducers to the analogous problem for multi-sequential transducers. 
Throughout, we assume that each finite-valued transducer is given as a disjoint union of finitely many unambiguous transducers. 
Given two such transducers, we first check if their domains are equal by testing the equivalence of their underlying nondeterministic automata, which is known to be $\CF{PSPACE}$-complete~\cite{stock73}. If their domains are not equal, then their distance is infinite. 
\begin{lemma}\label{finitetomulti}
    Let $\calT$ and $\calS$ be two finite-valued transducers with $dom(\calT)= dom(\calS)$. There exist multi-sequential transducers $\calT'$ and $\calS'$, effectively constructible from $\calT$ and $\calS$, s.t.~$\dom(\calT') = \dom(\calS')$ and $d(\calT, \calS) = d(\calT', \calS')$ for any metric $d$ given in \Cref{table:editdistance}.
\end{lemma}

\begin{proof}

Assume that we have two finite-valued transducers $\calT \equiv \calT_1 \cup \cdots \cup \calT_m$ and $\calS \equiv \calS_1 \cup \cdots \calS_n$, each expressed as the union of $m$ and $n$ state-disjoint unambiguous transducers, respectively. Let $L$ denote the common domain of $\calT$ and $\calS$. For $i \in [m]$ and  $j \in [n]$ let $\calT_i = (Q_i, A, B, s_i,\Delta_i,F_i, \lambda_i, \omicron_i)$, and $\calS_j = (P_j, A, B, t_j,\Gamma_j,G_j,\mu_j,\eta_j)$ where $\calT_i$ and $\calS_j$ are unambiguous transducers that are complete.

 We define $m+n$ sequential transducers $T_i'$ and $S_j'$, for $i \in [m]$ and $ j \in [n]$, whose underlying automaton is a product automaton $\calA$ that captures the \emph{synchronous runs} of all components of $T$ and $S$. The output labels of each $T_i'$ and $S_j'$ are defined depending on the outputs of $T_i$ and $S_j$, respectively. Before giving the formal definition of the underlying product automaton $\calA$, we give an intuitive overview of its construction. The input alphabet of $\calA$ is over the tuples of transitions of all components of $\calT$ and $\calS$, i.e., $\Delta_1 \times \cdots \times \Delta_m \times \Gamma_1 \times \cdots \times \Gamma_n$, and hence making the automaton deterministic. The states of $\calA$
track, for each component, (i) the state reached by the current run, and
(ii) the set of states reachable on the corresponding input word.
Final states are defined so that an accepting run of 
$\calA$ corresponds to a tuple of synchronous runs of all components of $\calT$ and $\calS$ on a common input word in their domain, s.t.~(i) at least one component of 
$\calT$ and at least one component of 
$\calS$ is accepting, and (ii) for every component whose reachable state set contains a final state, the state reached by the current run of that component is itself final. This condition ensures that all accepting runs induced by an input word are realised simultaneously by a single global run in  $\calA$.

Formally, $\calA = (Q,A', s, \Delta,F)$
where
\begin{itemize}
    \item $Q = Q'_1 \times \cdots \times Q'_m \times P'_1 \times \cdots P'_n$ is the set of states of $\calA$ s.t.~for $i \in [m]$ and  $j \in [n]$, $Q'_i = Q_i \times 2^{Q_i}$  and $P'_j = P_j \times 2^{P_j}$. Each \((q,M)\in Q'_i\) (similarly $P'_j$) keeps track of two things: 
the state \(q\) reached in $T_i$ (resp. $S_j$) on a run, and the set \(M\) of all states reachable in $T_i$ (resp. $S_j$) on the input word read along this run. 
    
    \item $A' = \Delta_1 \times \cdots \times \Delta_m \times \Gamma_1 \times \cdots \times \Gamma_n$ is the input alphabet of $\calA$.

    \item $s = ((s_1, \{s_1\}), \cdots, (s_m,\{s_m\}),(t_1,\{t_1\}), \cdots, (t_n,\{t_n\}))$ is the initial state of $\calA$.

    \item $\Delta \subseteq Q \times A' \times Q$ is the set of transitions of $\calA$ defined as follows. \\
  Let $\mathbf{q}  = \big( (q_1, M_1), \dots, (q_m, M_m), (p_1, N_1), \dots, (p_n, N_n) \big)$, $
    {\sigma} = (\delta_1, \dots, \delta_m, \gamma_1, \dots, \gamma_n)$, and $
    \mathbf{q'} = \big( (q'_1, M'_1), \dots, (q'_m, M'_m), (p'_1, N'_1), \dots, (p'_n, N'_n) \big)$. The transition 
$(\mathbf{q}, {\sigma}, \mathbf{q'}) \in \Delta$ iff there exists a letter $a \in A$ s.t.~the following conditions hold:
\begin{itemize}
    \item For each $i \in [m]$, $\delta_i = (q_i,a,q'_i) \in \Delta_i$ is a transition in $\calT_i$, and the  set 
    $
    M'_i = \{ q' \mid q \in M_i, (q,a,q') \in \Delta_i \}
    $
    consists of all states in $\calT_i$ reachable on an $a$ from the states in $M_i$.
    
    \item For each $j \in [n]$, $\gamma_j = (p_j,a,p'_j) \in \Gamma_j$ is a transition in $\calS_j$, and the  set 
    $N'_j = \{ q' \mid q \in N_j, (q,a,q') \in \Gamma_j \}$
    consists of all states in $\calS_j$ reachable on an $a$ from states in $N_j$.
\end{itemize}

\item $F$ is the set of all final states of $\calA$ s.t.~$\big( (q_1, M_1), \dots, (q_m, M_m), (p_1, N_1), \dots, (p_n, N_n) \big) \in F$ iff the following holds:
\begin{itemize}
\item there exist an $i\in [m]$ and $j\in [n]$ s.t.~$q_i \in F_i$ and $p_j \in G_j$.
    \item for all $i\in [m]$, if $M_{i}\cap F_{i} \neq \emptyset$, then $q_{i}\in F_{i}$.
    \item for all $j\in [n]$, if $N_{j}\cap G_{j} \neq \emptyset$, then $p_{j}\in G_{j}$.
\end{itemize}

\end{itemize}

By construction, the language $L' \subseteq (\Delta_1 \times \cdots \times \Delta_m \times \Gamma_1 \times \cdots \times \Gamma_n)^* $ of $\calA$ consists of all sequences of tuples of the form
 $\rho=(\delta_1^{1},\cdots,\delta_m^{1},\gamma_1^{1}, \cdots, \gamma_n^{1}) \cdots (\delta_1^{k},\cdots,\delta_m^{k},\gamma_1^{k}, \cdots, \gamma_n^{k}),\text{ for } k\geq 1$  s.t.~for each $\ell \in [k]$, $i \in [m]$, and $j \in [n]$, we have $\delta_i^{\ell} \in \Delta_i$ and $\gamma_j^{\ell} \in \Gamma_j$. Moreover, there exists a word $w = a_1 \cdots a_k$ (denoted by $w_{\rho}$) s.t.
\begin{enumerate}
\item For every $\ell \in [k]$, $i \in [m]$, and $j \in [n]$, the transitions $\delta_i^{\ell}$ and $\chi_j^{\ell}$ are taken on the same input letter $a_\ell$.
\item For all $i \in [m]$, the sequence $\delta^1_i, \ldots \delta^k_i$, is a valid run on $w$ in $T_i$. Similarly, for all $j \in [n]$, the sequence $\gamma^1_j, \ldots \gamma^k_j$, is a valid run on $w$ in $S_j$.
\item There exists $i \in [m]$ and $j \in [n]$, s.t.~the sequence $\delta_i^{1} \cdots \delta_i^{k}$ forms an accepting run of $\calT$ on $w$ and $\gamma_j^{1} \cdots \gamma_j^{k}$ forms an accepting run of $\calS$ on $w$. 

\item For all $i \in [m]$, if $w \in \dom(\calT_i)$, then  $\delta_i^{1} \cdots \delta_i^{k}$ forms the unique accepting run of $\calT_i$ on $w$, denoted as $\rho_{\text{\tiny$\calT$}_i}$. Similarly, for all $j \in [n]$, if $w \in \dom(\calS_j)$, then  $\gamma_j^{1} \cdots \gamma_j^{k}$ forms the unique accepting run of $\calS_j$ on $w$, denoted as $\rho_{\text{\tiny$\calS$}_j}$.
\end{enumerate}

Using the product automaton defined above, we now define two multi-sequential transducers  $\calT' \equiv \calT'_1 \cup \cdots \cup \calT'_m$ and $\calS' \equiv \calS'_1 \cup \cdots \cup \calS'_n$ as follows: for $i \in [m],\ j\in [n]$, let $\calT'_i = (Q,A',B, s, \Delta, F, \lambda'_i, \omicron'_i)$  and  $\calS'_j = (Q, A', B, s, \Delta,F,\mu'_j,\eta'_j)$ where  
\begin{itemize}
    \item $(Q,A',B, s, \Delta, F) = \calA$.
     \item For $i\in[m], j\in [n]$, $\lambda'_i: \Delta \rightarrow B^*$ and $\mu'_j: \Delta \rightarrow B^*$ are the output transition labelling functions of  of $\calT'_i$ and $\calS'_j$ respectively and are defined as follows.  If $\psi = (\mathbf{q},(\delta_1,\cdots,\delta_m,\gamma_1,\cdots,\gamma_n),\mathbf{q'}) \in \Delta$, then $\lambda'_i(\psi) = \lambda_i(\delta_i)$ and $\mu'_j(\psi) = \mu_j(\gamma_j)$.
    
    \item For all $i \in [m],j \in [n]$, $\omicron'_i:Q \rightarrow B^*$ and $\eta'_j:Q \rightarrow B^*$ are the output state labelling functions of $\calT'_i$ and $\calS'_j$ respectively and are defined as follows: \\
    \[\omicron'_i((q_1,M_1),\cdots,(q_m,M_m),(p_1,N_1),\cdots,(p_n,N_n)) = \omicron_i(q_i) \text{ and} \]
     \[\eta'_j((q_1,M_1),\cdots,(q_m,M_m),(p_1,N_1),\cdots,(p_n,N_n)) = \eta_j(q_j).\]
\end{itemize}

Each $\calT'_i$ (resp. $\calS'_j$) simply projects the $i$-th output of $\calT_i$ (resp. $j$-th output of $\calS_j$) along the synchronous run. Since both $\calT'$ and $\calS'$ have the same underlying automaton $\calA$, $\dom(\calT') = \dom(\calS') =\dom(\calA)$. It remains to show that $d(\calT,\calS) = d(\calT',\calS')$. Observe that there is a correspondence between $L$ and $L'$. For each $w \in L$, there is a $\rho \in L'$ (with $w_{\rho} = w$)  s.t.~$\calT(w) = \calT'(\rho)$ and $\calS(w) = \calS'(\rho)$. This follows since, by construction of final state of $\calA$, for all $i \in [m]$, if $w \in \dom(\calT_i)$ then $\rho_{\text{\tiny$\calT$}_i}$ is an accepting run of $\calT_i$ on $w$, and likewise, for all $j \in [n]$, if $w \in \dom(\calS_j)$ then $\rho_{\text{\tiny$\calS$}_j}$ is an accepting run of $\calS_j$ on $w$. Note that $\rho$ is not necessarily unique (since there could be some
non-determinism in a rejecting component leading to different sequences of transitions), however, the set of outputs remains the same for
$\rho$ and $w$. 
Similarly, for each $\rho \in L'$, there exists a unique $w \in L$ s.t.~$w=w_{\rho}$, and  $\calT(w) = \calT'(\rho)$ and $\calS(w) = \calS'(\rho)$. 
Thus, $\Big\{d\bigl(\calT(w),\calS(w)\bigr) \Big| w \in L\Big\} = \Big\{d\bigl(\calT'(\rho),\calS'(\rho)\bigr) \Big| \rho \in L'\Big\}$.
Hence, we conclude that  $d(\calT,\calS) = d(\calT',\calS')$.  \qedhere 
\end{proof}

\subsection{Edit distance of multi-sequential transducers}
In this subsection, we present the computability of the edit distance for multi-sequential relations, 
which form a subclass of finite-valued relations. This result will be used to address the computability of the edit distance for general finite-valued relations by virtue of \Cref{finitetomulti}. 

Our approach reduces the problem of computing the edit distance between multi-sequential 
relations (or multi-sequential transducers) to computing the \emph{relative distance} (see \Cref{def:relative}) between a 
sequential function and a multi-sequential relation. 
\begin{definition}[Relative distance]\label{def:relative}
    Let $d$ be a metric over words. The relative distance from a function $f$ to a relation $R$ w.r.t.~$d$, denoted by $\overrightarrow{d}(f,R)$, is defined as the least upper bound of the minimal distance between the output of $f$ and some output of $R$ on each input.
 $$\overrightarrow{d}(f,R) = \begin{cases} \sup \Big \{\inf \left \{d(f(u),v) \mid v \in R(u) \right \}  \Big | u \in \dom(f) \Big \} & \text{if $\dom(f)\subseteq \dom(R)$} \\
 \infty & \text{otherwise} 
 \end{cases}$$\end{definition}
The notion of relative distance is analogous to that of the directed distance defined for languages.
The following lemma shows how to compute the edit distance between two multi-sequential relations using the notion of relative distance. We assume that the multi-sequential relation is given by a finite union of sequential functions.
 \begin{lemma}\label{prop:multusingreldis}
     Let $R \equiv f_1 \cup \cdots \cup f_m$ and $S \equiv g_1 \cup \cdots \cup g_n $ be two multi-sequential relations where $f_i$ and $g_j$ are sequential functions for all $i \in [m]$ and $ j \in [n]$. Then, for any integer-valued metric $d$, 
      \[d(R,S) = \begin{cases} \max \Bigg( \ \left \{\overrightarrow{d}(f_i,S) \mid i \in [m] \right \}  \cup \left \{\overrightarrow{d}(g_j,R) \mid j \in [n] \right \} \ \Bigg) & \text{if $\dom(R)= \dom(S)$} \\
 \infty & \text{otherwise} 
 \end{cases}\]
 \end{lemma}
 \begin{proof}
      There are two cases to consider.
      \begin{description}
          \item[Case 1: $d(R,S) = \infty$.] Assume for contradiction that \[\max \Big(\left \{\overrightarrow{d}(f_i,S) \mid i \in [m] \right \}  \cup \left \{\overrightarrow{d}(g_j,R) \mid j \in [n] \right \} \Big) \leq k \text{ for some $k \geq 0$}.\] Thus, $\dom(R) = \dom(S)$. Moreover, by definition of relative distance, we can deduce that for any input $u$ and for any output $v \in R(u)$, there exists an output $v' \in S(u)$ s.t.~$d(v,v') \leq k$, and vice-versa. By \Cref{def:distrel}, we get that $d(R,S) \leq k$, which is a contradiction.

          \item[Case 2: $d(R,S) = k$ for some $k\geq 0$.] For any input $u$ and any output $v \in R(u)$, there exists an output $v' \in S(u)$ s.t.~$d(v,v') \leq k$, and vice-versa. Hence, for all $i \in [m]$, $\overrightarrow{d}(f_i,S) \leq k$ and for all $j \in [n]$, $\overrightarrow{d}(g_j,R) \leq k$. Since $d(R,S) = k$, there exist an input $u$, s.t.~$H_d(R(u),S(u)) = k$. This implies that there exists an $i \in [m], j \in [n]$ s.t.~either $\overrightarrow{d}(f_i,S) =k$ or $\overrightarrow{d}(g_j,R)=k$. Hence, $\max \Big( \left \{\overrightarrow{d}(f_i,S) \mid i \in [m] \right \}  \cup \left \{\overrightarrow{d}(g_j,R) \mid j \in [n] \right \} \Big) =k= d(R,S)$. 
          
          This completes the proof of the lemma. \qedhere
      \end{description}
       \end{proof}
We introduced the notion of \emph{relative distance} since the distance $d(R,S)$ cannot be expressed in terms of $d(f_i,S)$ and $d(g_j,R)$ (instead of $\overrightarrow{d}(f_i,S)$ and $\overrightarrow{d}(g_j,R)$). For instance, consider $R \equiv f\cup g$, where $f,g : A^* \rightarrow B^*$ are sequential functions defined as: $f(w)=a^{|w|}$, $g(w)=b^{|w|}$ for all $w \in A^*$. Since $d$ is a metric, $d(R,R)=0$. However, both $d(f,R)$ and $d(g,R)$ are infinite. Note that the distance between a function $f$ and a relation $R$ is finite iff, for every input $u$, the distance between $f(u)$ and \emph{all} outputs in $R(u)$ are bounded.

The next section is dedicated to showing that the relative distance between a sequential function and a multi-sequential relation with respect to metrics $d$ given in \Cref{table:editdistance} is computable (see \Cref{th:relativeDist}). As a consequence, we obtain the following result using  \Cref{finitetomulti} and \Cref{prop:multusingreldis}.

\begin{theorem}
\label{th:finitevalued} 
The distance between two finite-valued rational relations with respect to any metric given in \Cref{table:editdistance} is computable.
\end{theorem}

\section{Computing relative distance}
\label{sec:relativeDistance}
We show that the problem of computing the relative distance of a function $f$ to a relation $R$ can be reduced to two boundedness problems: (1) deciding whether it is finite \textit{(finiteness)} and, (2) if finite, whether it is bounded by a given constant $k$ \textit{($k$-finiteness)}. This formulation, along with its proof, closely follows Proposition 3.6 in \cite{editdistance}.
 \begin{proposition}\label{prop:disttoclose}
    Let $d$ be an integer-valued metric, $f$ a function, and $R$ a relation. The relative distance $\overrightarrow{d}(f,R)$ is computable iff it is decidable whether\\ 
    {\upshape\textbf{(1)}} $\overrightarrow{d}(f,R) < \infty$, \qquad and \qquad {\upshape\textbf{(2)}}  $\overrightarrow{d}(f,R) \leq k$ for some $k\geq 0$.
        
\end{proposition}
\begin{proof}
     Clearly, if the relative distance $\overrightarrow{d}(f,R)$ with respect to $d$ is computable, then we can decide whether $\overrightarrow{d}(f,R) < \infty$ and whether $\overrightarrow{d}(f,R) \leq k$ for some $k \geq 0$. 
     For the converse, assume decidability of the two boundedness problems. Given a function $f$ and a relation $R$, we first check whether  $\overrightarrow{d}(f,R) < \infty$. If it is not, then $\overrightarrow{d}(f,R) = \infty$. Otherwise, we perform an exponential search:  check whether  $\overrightarrow{d}(f,R) \leq k$ for $k=2^0, 2^1, 2^2,\ldots$ until the condition fails. Once an interval containing the relative distance is found, we perform a binary search on the interval $[2^n, 2^{n+1}]$, $n\in \N$ to determine the exact value of $\overrightarrow{d}(f,R)$. 
\end{proof}

We show in the upcoming subsections that both finiteness and $k$-finiteness of relative distance of a sequential function to a multi-sequential relation w.r.t.~Levenshtein family of distances are decidable (See Lemma~\ref{lem:compreldis} and \ref{lem:kclosemult}). Hence, using \Cref{prop:disttoclose}, we get
\begin{theorem}\label{th:relativeDist}
    Let $f$ be a sequential function, and $R$ be a multi-sequential relation. The relative distance of a $\overrightarrow{d}(f,R)$ is computable for $d \in \{\dlev, d_{lcs}, d_{dl}\}$.
\end{theorem}

\subsection{Deciding finiteness of relative distance}
We will now show that given a sequential function $f$ and a multi-sequential relation $R$, the problem of deciding whether $\overrightarrow{d}(f,R) < \infty$ is decidable. The first step is to check whether $\dom(f) \subseteq \dom(R)$; this reduces to the language inclusion problem of their underlying automata, which is decidable~\cite{stock73}. If $\dom(f) \not\subseteq \dom(R)$, then $\overrightarrow{d}(f,R) = \infty$.

\subparagraph*{Partitioning the domain of $f$.} Now, consider the case where $\dom(f) \subseteq \dom(R)$. Assume that the multi-sequential relation $R$ is given as the union of $m$ (complete) sequential transducers, $D_1, \ldots D_m$. We partition $\dom(f)$ into $2^{m} -1$ classes, where each class is labelled by a nonempty subset of $[m]$. We denote the set of classes in the partition by $\calP = \{ P \subseteq [m] \mid P \neq \emptyset\}$. For each $P \subseteq [m]$, the corresponding class $C_P$ contains exactly those words $w\in\dom(f)$ such that $w\in\dom(D_i)$ for every $i\in P$, and $w\notin\dom(D_j)$ for every $j\in [m]\setminus P$. 
The set $\{C_P\}_{P\in\mathcal P}$ forms a partition of $\dom(f)$, i.e., $\dom(f)=\biguplus_{P\in\mathcal P} C_P$. Let $f_{|P}$ denote the function $f$ restricted to $C_P$. It is straightforward to observe that 
\begin{equation}\label{eq: rel}
\overrightarrow{d}(f,R) = \max \{\overrightarrow{d}(f_{|P},R) \mid P \in \calP\}.
\end{equation}
For checking whether $\overrightarrow{d}(f,R) < \infty$, it now suffices to check whether $\overrightarrow{d}(f_{|P},R) < \infty$ for each $P \in \calP$. 
Finiteness of relative distance is solved using structural arguments based on conjugacy and SCC decompositions. Since 
$R$ is multi-sequential, different input words may belong to the domain of different component transducers $D_i$. Fixing a partition $P$ allows us to restrict the structural arguments (particularly Claim~\ref{claim:connected} and \ref{claim:distclaim}) to only those components in $P$ that accept the set of input words $C_P$ within the domain of $f$.

\subparagraph*{Checking whether $\overrightarrow{d}(f_{|P},R) < \infty$.} Fix a $P \in \calP$. We construct a \emph{product transducer} $\calT$, defined as the Cartesian product of the transducers given by $f$ and $R$. Let the sequential function $f$ be defined by a (complete) sequential transducer 
$
D_0 = (Q_0, A, B, \Delta_0, s_0, F_0, \lambda_0, \omicron_0),
$ 
and let the multi-sequential relation $R$ be given as the union of $m$ (complete) sequential transducers 
$
D_i = (Q_i, A, B, \Delta_i, s_i, F_i, \lambda_i, \omicron_i) \quad \text{for } i \in [m], \; m \in \mathbb{N}.
$
The product transducer
$\calT$ is defined as follows, 
\[ \calT = (Q, A, B', \Delta, s, F, \lambda, \omicron), \text{ where} \] 
\begin{itemize}
    \item $Q = Q_0 \times Q_1 \times \cdots \times Q_m$ is the set of states,
    \item $B' = \overbrace{(B \cup \{\varepsilon\}) \times \cdots \times (B \cup \{\varepsilon\})}^{\text{$m+1$ times}}$ is the output alphabet, 
    \item $s= (s_0,s_1,\cdots, s_m)$ is the initial state,
    \item $ \Delta= \big\{\big((q_0,q_1,\cdots,q_m),a,(p_0,p_1, \cdots,p_m)\big) \mid \forall i \in \{0, 1,\cdots, m\},\  (q_i,a,p_i) \in \Delta_i\big\}$  is the set of transitions,
    \item $F= \{(q_0,q_1,\cdots,q_m) \mid q_0 \in F_0\text{ and }  \forall i\in [m],\ q_i \in F_i \iff i \in P \}$  is the set of final states,
    \item $\lambda$ is the output transition labelling function where 
        $\lambda((q_0,q_1,\cdots,q_m),a,(p_0,p_1,\cdots,p_m)) = (v_0,v_1,\cdots,v_m)$
        if $v_i = \lambda_i(q_i,a,p_i) $ for $i \in \{0,1,\cdots,m\}$, and 
    \item $\omicron$ is the output state labelling function and is defined as \[\omicron(q_0,q_1,\cdots,q_m) = (\omicron_0(q_0),\omicron_1(q_1),\cdots,\omicron_m(q_m)).\]
\end{itemize}

Note that the constructed product transducer $\calT$ is sequential and $\dom(\calT) = \dom(f_{|P}) = C_P$. WLOG, we further assume that $\calT$ is trim for the remainder of our proofs. 

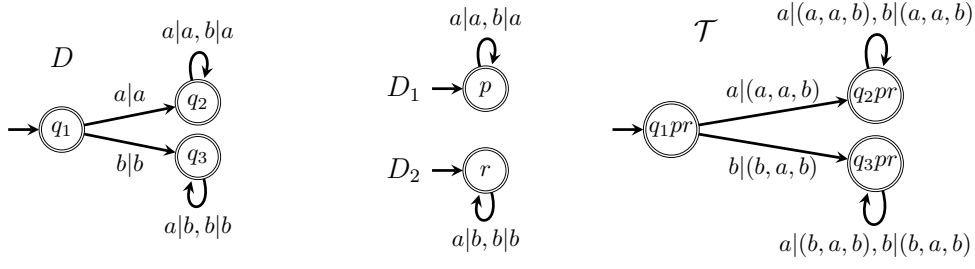
\begin{figure}
\centering
\scalebox{.9}{
\begin{tikzpicture}[>=stealth,auto]
\tikzset{
state/.style={circle,draw,minimum size=18pt,inner sep=1pt},
box/.style={draw,thick,gray,minimum width=5.2cm,minimum height=3.5cm},
every edge/.style={draw,->,line width=.4mm}
}
\tikzset{initial text={}}

\node[box,anchor=west,color=white] (B1) at (0,0) {};
\node[box,anchor=west,,color=white] (B2) at (B1.east) {};
\node[box,anchor=west,,color=white] (B3) at (B2.east) {};

\node[state,initial,accepting] (q1) at ($(B1.center)+(-1,0)$) {$q_1$};
\node[state,accepting] (q2) at ($(B1.center)+(1,0.4)$) {$q_2$};
\node[state,accepting] (q3) at ($(B1.center)+(1,-0.4)$) {$q_3$};

\path
(q2) edge[loop above] node {$a|a,b|a$} ()
(q1) edge node[above] {$a|a$} (q2)
(q3) edge[loop below] node {$a|b,b|b$} ()
(q1) edge node[below] {$b|b$} (q3);

\node at ($(B1.north)+(-1,-0.7)$) {\Large $D$};

\node[state,initial,accepting] (p) at ($(B2.center)+(0,0.6)$) {$p$};
\node[state,initial,accepting] (r) at ($(B2.center)+(0,-0.6)$) {$r$};

\path
(p) edge[loop above] node {$a|a,b|a$} ()
(r) edge[loop below] node {$a|b,b|b$} ();

\node at ($(B2.center)+(-1.2,0.6)$) {\Large $D_1$};
\node at ($(B2.center)+(-1.2,-0.6)$) {\Large $D_2$};

\node[state,initial,accepting] (s1) at ($(B3.center)+(-2.5,0)$) {$q_1pr$};
\node[state,accepting] (s2) at ($(B3.center)+(.5,0.5)$) {$q_2pr$};
\node[state,accepting] (s3) at ($(B3.center)+(.5,-0.5)$) {$q_3pr$};

\path
(s2) edge[loop above] node {$a|(a,a,b),b|(a,a,b)$} ()
(s1) edge node[above] {$a|(a,a,b)$} (s2)
(s3) edge[loop below] node {$a|(b,a,b),b|(b,a,b)$} ()
(s1) edge node[below] {$b|(b,a,b)$} (s3);

\node at ($(B3.north)+(-2,-0.3)$) {\Large $\calT$};

\end{tikzpicture}}
\caption{A sequential transducer $D$ recognising the function 
$f:\,\sigma u \mapsto \sigma^{|u|+1}$; a multi-sequential transducer 
$D_1 \cup D_2$ recognising the relation $g: u \mapsto \{a^{|u|}, b^{|u|}\}$; 
and their product transducer $\calT$ w.r.t.~unique partition $P = \{1,2\}$ where $C_P = \dom(f)$.}
\label{fig}
\end{figure}

The idea is to reduce the problem of checking whether $\overrightarrow{d}(f_{|P},R) < \infty$ to several instances of checking whether the distance between two sequential transducers is finite. Checking finiteness of distance for functional transducers w.r.t.~Levenshtein family was based on the notion of conjugacy in loops (See Propositions~\ref{prop:conj} and \ref{prop:levboundconj}). Conjugacy in loops is also necessary in the context of relative distance, as stated in the following claim. Towards this, we define connected loops. Two loops of $\calT$, rooted at states $q_1$ and $q_2$ respectively, are \emph{connected} if there exists a path between $q_1$ and $q_2$ in $\calT$.  Extending this notion to sets of loops, we say that a set of loops is connected in $\calT$ if there exists a run of $\calT$ such that all loops in the set appears in the run.

\begin{claim}\label{claim:connected}
Let $L$ be a set of connected loops in $\calT$. For each loop $l \in L$, let $(v_0^l, v_1^l, \ldots, v_m^l)$ denote the corresponding output tuple produced along that loop. If $\overrightarrow{d}(f_{|P},R) < \infty$, then there exists a common index $i \in P$ s.t.~$v_0^l$ and $v_i^l$ are conjugate for every $l \in L$.
\end{claim}

\begin{claimproof}
Assume $\overrightarrow{d}(f_{|P},R) = K$ for a constant $K \geq 0$. Let $l_1$ and $l_2$ be two connected loops rooted at states $q_1$ and $q_2$, respectively. Without loss of generality, assume that there is a path from $l_1$ to $l_2$ in $\calT$. Since $\calT$ is trim, there exist  words $u,v,u',v',w$, an accepting state $q_f$ and some output tuples \\
$
(u_0, u_1, \ldots, u_m), \; (v_0, v_1, \ldots, v_m), \; (u'_0, u'_1, \ldots, u'_m), \; (v'_0, v'_1, \ldots, v'_m), \; (w_0, w_1, \ldots, w_m),
$ such that
\[
s \xrightarrow{u | (u_0, \ldots, u_m)} q_1 \xrightarrow{v | (v_0,\ldots,v_m)} q_1 \xrightarrow{u' | (u'_0, \ldots, u'_m)} q_2 \xrightarrow{v' | (v'_0,\ldots,v'_m)} q_2 \xrightarrow{w|(w_0,\ldots,w_m)} q_f.
\]

Since $\overrightarrow{d}(f_{|P},R) =K $, for every input word in $\dom(f_{|P})$, the edit distance between the output of $f_{|P}$ and some output of $R$ is less than or equal to $K$. In particular, for each input of the form $uv^{k_1}u'{v'}^{k_2} w$ with $k_1,k_2 \ge 0$, there exists an index $j \in P$ s.t.~
$
d\big(u_0 v_0^{k_1} u'_0 {v'_0}^{k_2} w_0,\; u_j v_j^{k_1} u'_j {v'_j}^{k_2} w_j\big) \le K$. We restrict only to indices in $P$ since the input word in $\dom(f_{|P}) = C_P$ will not be accepted by any transducer $D_j$ for $j \in [m] \setminus P$ (by definition of the partition $P$).

\medskip
\noindent
\textbf{Finding indices conjugate to $v_0$ and $v_0'$ in $P$.}  
First, increase $k_1$ while keeping $k_2=1$. By the pigeonhole principle, since $|P|$ is finite, there exists an index $i_1 \in P$ and an infinite subset $I \subseteq \mathbb{N}$ s.t.~for all $k_1 \in I$,
$d(u_0 v_0^{k_1} u'_0 {v'_0} w_0, u_{i_1} v_{i_1}^{k_1} u'_{i_1} {v'_{i_1}} w_{i_1}) \le K$.
Applying \Cref{prop:conj}, we conclude that $v_0$ is conjugate to $v_{i_1}$. Let $J_1 = \{ j \in P \mid v_0 \sim v_j \}$ denote the set of all indices in $P$ whose outputs are conjugate to $v_0$.
By a symmetric argument by fixing $k_1=1$ and increasing $k_2$,  we get the set $J_2 = \{ j \in P \mid v_0' \sim v_j' \}$, the set of indices in $P$ whose outputs are conjugate to $v_0'$.

\medskip
\noindent
\textbf{Existence of a common index in $P$.}  
Assume for contradiction that there is no common index $i \in P$ such that $v_0 \sim v_i$ and $v_0' \sim v_i'$, i.e., $J_1 \cap J_2 = \emptyset$. We show that, for all $i \in P$, $d(u_0 v_0^{k_1} u'_0 {v'_0}^{k_2} w_0, u_i v_i^{k_1} u'_i {v'_i}^{k_2} w_i) > K$ when $k_1, k_2$ are sufficiently large (say $2^K$). Since $J_1 \cap J_2 = \emptyset$, for any $i \in [m]$, exactly one of the following three cases holds:
{\renewcommand{\labelenumi}{\textbf{Case-\arabic{enumi}:}}
\setlength{\leftmargini}{40pt}
\begin{enumerate}
    \item $i \in J_1$, $i \notin J_2$. Here $v_0 \sim v_i$ but $v_0' \not\sim v_i'$. So for sufficiently large $k_2$, the edit distance
    $
    d(u_0 v_0^{k_1} u'_0 {v'_0}^{k_2} w_0, u_i v_i^{k_1} u'_i {v'_i}^{k_2} w_i)
    $ 
    grows unboundedly (\Cref{prop:conj}), contradicting\\ $d(u_0 v_0^{k_1} u'_0 {v'_0}^{k_2} w_0, u_i v_i^{k_1} u'_i {v'_i}^{k_2} w_i) \leq K$.

    \item $i \notin J_1$, $i \in J_2$. Here $v_0 \not\sim v_i$ but $v_0' \sim v_i'$. So for sufficiently large $k_1$, the distance
    $
    d(u_0 v_0^{k_1} u'_0 {v'_0}^{k_2} w_0, u_i v_i^{k_1} u'_i {v'_i}^{k_2} w_i)
    $ 
    becomes unbounded, again contradicting\\ $d(u_0 v_0^{k_1} u'_0 {v'_0}^{k_2} w_0, u_i v_i^{k_1} u'_i {v'_i}^{k_2} w_i) \leq K$.

    \item $i \notin J_1 \cup J_2$. Here $v_0 \not\sim v_i$ and $v_0' \not\sim v_i'$, for sufficiently large $k_1$ and $k_2$, the distance becomes unbounded contradicting $d(u_0 v_0^{k_1} u'_0 {v'_0}^{k_2} w_0, u_i v_i^{k_1} u'_i {v'_i}^{k_2} w_i) \leq K$. 
\end{enumerate}
}
Since all three cases lead to a contradiction, there must exist a common index $i \in P$ such that $v_0 \sim v_i$ and $v_0' \sim v_i'$. This proof can be generalised to any set of loops that occur along a single run, that is, to any set of connected loops.
\end{claimproof}

  \Cref{claim:connected} holds when the loops are connected and need not be true for all loops of a trim transducer $\calT$. For instance, the loops in the product transducer $\calT$, depicted in \Cref{fig}, rooted at states $q_{2}pr$ and $q_{3}pr$ are disjoint and have no common index witnessing conjugacy: in the former loop the outputs of $\calD$ and $\calD_1$ are conjugate, whereas in the latter they are not.

     Now, in order to decide whether $\overrightarrow{d}(f_{|P},R) < \infty$, we decompose $\calT$ into a directed acyclic graph of maximal strongly connected components (SCCs) $S_1,\dots,S_r\subseteq Q$ for some $r\in\mathbb{N}$ disregarding the labels on the transitions. 
     Consider the set of paths $\Pi$, where each $\pi \in \Pi$ is of the form $S_{i_1}t_{i_1}S_{i_2}\dots t_{i_{n-1}}S_{i_n}$ where $S_{i_1}$ is an SCC that contains an initial state, $S_{i_n}$ is an SCC that contains a final state, and for all $1\leq k<n$, $t_{i_k}$ is a transition of $\calT$ from a state in $S_{i_k}$ to some state in $S_{i_{k+1}}$. Let $\calT_\pi$ denote the trim subtransducer of $\calT$ obtained by removing all the transitions in $\calT$ except the transitions $t_{i_k}$ $(1\leq k<n)$ and the transitions within the SCCs $S_{i_k}$ for $k\in[n]$. 
      Note that since the SCCs are maximal, the set $\Pi$ is finite. 
    Now, it is straightforward to see that $\calT \equiv \bigcup_{\pi \in \Pi} \calT_\pi$. 
    
    For each $i \in \{0, \ldots, m\}$, let $\mathcal{T}_{\pi(i)}$ denote the transducer obtained from $\mathcal{T}_\pi$ by projecting the output labelling functions to the $i$-th component in the output of $T_{\pi}$. That is, a transition $\delta$ in $\mathcal{T}_\pi$ labelled with the output tuple $(v_0, v_1, \ldots, v_m)$ is labelled with $v_i$ in $\mathcal{T}_{\pi(i)}$. Similarly, a state $q$ in $\mathcal{T}_\pi$ labelled with the output tuple $(v_0, v_1, \ldots, v_m)$ is labelled with $v_i$ in $\mathcal{T}_{\pi(i)}$. Note that $\dom(\calT_{\pi(i)}) = \dom (\calT)$.

    \begin{claim}\label{claim:distclaim}
       $\overrightarrow{d}(f_{|P},R) < \infty$ iff for all paths $\pi \in \Pi$, $\exists i \in P$ s.t.~$d(\calT_{\pi(0)},\calT_{\pi(i)}) < \infty$.
    \end{claim}
    \begin{claimproof}
        ($\Leftarrow$) Assume that for each path $\pi \in \Pi$, there is an index $i_\pi \in P$ s.t.~$d(\calT_{\pi(0)}, \calT_{\pi(i_\pi)}) = k_\pi$ for some $k_\pi \geq 0$.  
        Since $\calT$ is sequential and $\calT \equiv \bigcup_{\pi \in \Pi} \calT_\pi$, every input word $w$ is accepted by a unique path $\pi_w \in \Pi$. In fact, $w$ is accepted by $\calT_{\pi_w(i)}$ for all $i \in P$ since $\dom(f_{|P}) = \dom(\calT) = C_P$. Thus, for any $w$, the edit distance between $f(w)$ and some output of $R(w)$ is bounded by $k_{\pi_w}$. Hence,
        $
        \overrightarrow{d}(f_{|P},R) \leq \max\{k_\pi \mid \pi \in \Pi\} 
        $.
        
        $(\Rightarrow)$ Assume $\overrightarrow{d}(f_{|P},R) < \infty$.  A path $\pi \in \Pi$ is of the form $S_1 t_1 S_2 t_2 \cdots S_n$, where $S_1$ contains an initial state, $S_n$ contains a final state, and every $S_j$ ($j \in [n]$) is strongly connected. Hence, all loops in $\pi$ are connected.  By \Cref{claim:connected}, there exists an index $i \in P$ s.t.~for every output tuple $(v_0, v_1, \ldots, v_m)$ produced along a loop of $\calT_\pi$, the words $v_0$ and $v_i$ are conjugate. This implies that the Cartesian product of $\calT_{\pi(0)}$ and $\calT_{\pi(i)}$ produces only conjugate pairs of output words. Also, $\dom(\calT_{\pi(0)}) = \dom(\calT_{\pi(i)}) = \dom (\calT)$. By \Cref{prop:levboundconj}, we get $d(\calT_{\pi(0)}, \calT_{\pi(i)}) < \infty$.   
        \end{claimproof}

    By virtue of the above claim, checking whether $\overrightarrow{d}(f_{|P},R) < \infty$ reduces to computing $d\bigl(\mathcal{T}_{\pi(0)}, \mathcal{T}_{\pi(i)}\bigr)$ for all $i \in P$ and all $\pi \in \Pi$, and verifying that for each path $\pi \in \Pi$, there exists some $i \in P$ s.t.~$d\bigl(\mathcal{T}_{\pi(0)}, \mathcal{T}_{\pi(i)}\bigr) < \infty$. Since  $\mathcal{T}$ is sequential, each $\mathcal{T}_{\pi(i)}$ for $i \in \{0, 1, \ldots, m\}$ and $\pi \in \Pi$ is also sequential. Therefore, computing $d\bigl(\mathcal{T}_{\pi(0)}, \mathcal{T}_{\pi(i)}\bigr)$ amounts to computing the edit distance between sequential transducers, which is decidable \cite{editdistance}. Consequently, using \Cref{eq: rel}, we get the following Lemma.

\begin{lemma}\label{lem:compreldis}
     Given a sequential function $f$ and a multi-sequential relation $R$, it is decidable whether the relative distance $\overrightarrow{d}(f,R)$ is finite for any metric $d \in \{\dlev, d_{lcs}, d_{dl}\}$.
 \end{lemma}

\subsection{Deciding k-finiteness of relative distance}
In the next lemma, we prove that the $k$-finiteness of relative distance is decidable. The proof is a generalisation of the proof used for checking whether the edit distance between two sequential functions is at most a given $k$ (Proposition 3.11, \cite{editdistance}).

\begin{lemma}\label{lem:kclosemult}
    For a given $k$, it is decidable to check whether the relative distance between a sequential function and a multi-sequential relation is less than or equal to $k$ with respect to any metric $d$ given in \Cref{table:editdistance}.
\end{lemma}

\begin{proof}
    Given a sequential function $f$ and a multi-sequential relation $R$ over the output alphabet $B$, we begin by checking whether $\dom(f) \subseteq \dom(R)$; if so, we construct the product transducer $\calT=(Q, A, B', s, \Delta, F, \lambda, \omicron)$, whose domain is $\dom(f) \cap \dom(R)$,  and on each input word, produces an $(m+1)$-tuple consisting of the output of $f$ together with the outputs of the $m$ sequential transducers whose union defines $R$. Let $L$ be the domain of $\calT$. For each $i \in \{0,1,\ldots,m\}$, let $\calT_i=(Q, A, B', s, \Delta, F, \lambda_i, \omicron_i)$  denote the transducer obtained from $\calT$ by projecting the output labelling functions to the $i$-th component of the $(m+1)$-tuple. i.e., for all $q,q'\in Q$ and $a\in B$, if $\lambda(q,a,q')= (v_0, v_1, \ldots, v_m)$, then $\lambda_i(q,a,q')= v_i$. Similarly, for all $q\in Q$, if $\omicron(q)= (v_0, v_1, \ldots, v_m)$, then $\omicron_i(q)=v_i$.

    Observe that $\overrightarrow{d}(f,R) \leq k$ for $k\in\mathbb{N}$ iff for all $w \in L$ we can perform at most $k$ edits to $\calT_0(w)$ and obtain $\calT_i(w)$ for some $i \in [m]$. 
    
    First, we check whether $\overrightarrow{d}(f,R) < \infty$. If so, there exists a maximum length difference, denoted by $\partial_{max}$, between the partial outputs of $\calT_0$ and $\calT_i$ for some $i \in [m]$ on any input. In fact,  $\partial_{max}= N \cdot \ell$, where $N$ is the number of states in $\calT$ and $\ell$ is the maximum length difference between outputs of $\calT_0$ and $\calT_i$ (for all $i \in [m]$) along any transitions.
    Assume for contradiction that there exists a partial input $u$ with output tuple $(v_0,\ldots,v_m)$ s.t.$\mod(|v_0|-|v_i|) > \partial_{max}$ for all $i \in [m]$. Since each transition can contribute at most $\ell$, the length difference on $u$ between outputs of $\calT_0$ and any $\calT_i$ is bounded by $|u|\ell$. If $|u|\ell > N\ell$, then $|u| > N$, and some state repeats along the run of $u$ in $\calT$. For each $i \in [m]$, there exists a loop where the length difference between the outputs of $\calT_0$ and $\calT_i$ increases; otherwise, it contradicts the fact that$\mod(|v_0|-|v_i|) > N \ell$. 
   The idea then is to show that pumping all the loops encountered during the run of $u$ in $\calT$ will strictly increase the length difference between the outputs of $\calT_0$ and $\calT_i$ for all $i\in[m]$ 
   contradicting $\overrightarrow{d}(f,R) < \infty$.

    For each metric $d \in \{d_l, d_{lcs}, d_{dl}\}$, fix the corresponding set of allowed edits $C$. In order to determine whether the relative distance is at most $k$, we construct a nondeterministic finite state automaton $\calA_{C,k}$. The automaton reads words $w \in L$ and accepts $w$ if there exists $i\in[m]$ s.t.~the output $\calT_i(w)$ can be obtained from $\calT_0(w)$ using at most $k$ edits from $C$. The NFA $\calA_{C,k} = (Q', A, s', \Delta', F')$ is defined as follows.
\begin{itemize}
\item  The states $Q'$ of $\calA_{C,k}$ are tuples $(q, (b_1,\ldots,b_m), (L_1,\ldots,L_m))$, where $q$ is a state of $\calT$; for each $i\in[m]$, $b_i \in \{0,\ldots,k\}$ denotes the
remaining edit budget associated with $\calT_i$; and $L_i$ records the current unmatched leftover between the outputs of $\calT_0$ and $\calT_i$.
Formally, for each each $i \in [m]$, a \emph{leftover} $L_i$ is either an element of $(B^* \times \{\epsilon\}) \cup (\{\epsilon\} \times B^*)$ or the special symbol $\text{x}$.  If $L_i = (u_i, \epsilon)$, then $u_i$ is the unmatched suffix of the output of $\calT_0$ that has not yet been matched with the output of $\calT_i$.  If $L_i = (\epsilon, u_i)$, then $u_i$ is the unmatched suffix of the output of $\calT_i$ that has not yet been matched with the output of $\calT_0$. Since the edits in $C$ are local, the alignment must eventually match a common prefix, leaving an unmatched suffix on only one side. If $L_i = \text{x}$, it indicates that the length of the unmatched output exceeds the threshold $\max(\partial_{max}, k)$ and can no longer be matched.  We restrict to states for which at least one $L_i \neq \text{x}$.

\item The initial state is $s' = (s, (k,\ldots,k), ((\epsilon,\epsilon),\ldots,(\epsilon,\epsilon)))$, where $s$ is the initial state of $\calT$.

\item The transition
$((q,(b_1,\ldots,b_m),(L_1,\ldots,L_m)), \sigma, (q',(b_1-b'_1,\ldots,b_m-b'_m),(L'_1,\ldots,L'_m)))$ is in $\Delta'$ 
if there exists a transition $\delta=(q,\sigma,q')$ of $\calT$ with
outputs $\lambda_0(\delta)$ for $\calT_0$ and $\lambda_i(\delta)$ for $\calT_i$ such that, for every $i\in[m]$, the following holds.
The automaton nondeterministically chooses $b'_i\le b_i$ and attempts a partial match using $b'_i$ edits from $C$:
\begin{itemize}
    \item if $L_i=(\epsilon,u)$, then $(\lambda_0(\delta),u\lambda_i(\delta))$
    is partially matched using $b'_i$ edits;
    \item if $L_i=(u,\epsilon)$, then $(u\lambda_0(\delta),\lambda_i(\delta))$
    is partially matched using $b'_i$ edits.
\end{itemize}
The unmatched suffix resulting from this partial match determines the new leftover $L'_i$; if its length exceeds $\max(\partial_{\max},k)$,
then $L'_i=\text{x}$, otherwise $L'_i$ is the resulting unmatched pair. If $L_i=\text{x}$ or no suitable $b'_i\le b_i$ exists, then we set
$b'_i=0$ and $L'_i=\text{x}$. The transition updates the edit budgets to
$(b_1-b'_1,\ldots,b_m-b'_m)$ and the leftovers to $(L'_1,\ldots,L'_m)$.

\item The set of accepting states is defined as
$F' = \{(q, (b_1,\ldots,b_m), (L_1,\ldots,L_m)) \in Q_A \mid q$ { is a final state of } $\calT$ { and there exists at least one } $L_i =(u,v) \neq \text{x} \text{ with }b_i \geq 0 \text{ and }$ $ d(u \omicron_0(q_f), v \omicron_i(q_f)) \le b_i\}$.

\end{itemize}

Note that for every word $w$ accepted by the NFA $\calA_{C,k}$, there exists an index $i \in [m]$ such that $L_i \neq \text{x}$, $b_i \ge 0$, and $ d(u \omicron_0(q_f), v \omicron_i(q_f)) \le b_i\}$. This implies that the distance between the outputs of $\calT_0$ and $\calT_i$ for the input $w$ is at most $k$.
This bound is witnessed by the sequence of edit operations encoded along the accepting run of $\calA_{C,k}$. 
Consequently, $\calA_{C,k}$ accepts exactly those words in the domain of $\calT$ for which the distance between $\calT_0$ and $\calT_i$ is bounded by $k$ for some $i \in [m]$. Therefore, to determine whether the relative distance between the transducers is at most $k$, it suffices to check whether $\dom(\calT) = \dom(\calA_{C,k})$. In other words, the relative distance between $f$ and $R$ with respect to $C$ is at most $k$ if and only if $L(\calA_{C,k}) = L$.
\end{proof}

\section{Discussion and conclusion}
\label{sec:conclusion}

In this work, we showed that the edit distance between finite-valued transducers is computable, extending the class of transducers for which edit distance computation is known to be decidable. The algorithm we present establishes decidability, but it is not optimised for efficiency. On analysing the algorithm, the first step of reducing the edit distance problem for finite-valued transducers (given by union of functional transducers) to multi-sequential transducers uses exponential space. Computing the edit distance for the resulting multi-sequential transducers then involves multiple instances of relative distance computation. Each instance requires deciding $k$-finiteness, which incurs another exponential blow-up in space, and checking the finiteness of the relative distance through a product automaton and distance computation for functional transducers, which requires doubly exponential space. Overall, these steps result in a 3-$\CF{EXPSPACE}$ procedure. Regarding the lower bound, computing the edit distance between two finite-valued transducers is at least as hard as checking their equivalence. This is because checking whether two transducers are equivalent reduces to computing their edit distance and verifying that it is zero. This problem is known to be PSPACE-hard. This follows from the fact that equivalence checking for nondeterministic automata — known to be PSPACE-complete~\cite{stock73}— can be reduced to equivalence checking for nondeterministic functional (i.e., 1-valued) transducers. Even for functional transducers, no better lower bound than PSPACE-hardness is known for computing the edit distance, while the best known upper bound is 2-EXPSPACE. In our setting, we rely on this result to compute the edit distance of multi-sequential transducers and incur an additional EXPSPACE overhead when reducing finite-valued transducers to the multi-sequential case. Improving the efficiency of this approach remains an important direction for future work.

The computation of relative distance for the metrics in \Cref{table:editdistance} relies on a characterisation connecting finiteness of edit distance with conjugacy (see Propositions~\ref{prop:conj} and~\ref{prop:levboundconj}). This property does not extend to all edit metrics; for instance, it fails for Hamming distance, where only substitutions are allowed. Nevertheless, all results except \Cref{lem:compreldis} remain valid for Hamming distance. \Cref{lem:compreldis} can be adapted using existing characterisations of Hamming distance (see Theorem~4.10 in \cite{editdistance} and Lemma~3.6 in \cite{approximate}). More generally, the method for computing relative distance depends on the set of allowed edit operations, and extending our techniques to more general distance notions is a natural direction for future work.

Promising directions for future work include improving the complexity of our constructions, extending to other distance measures, and adapting these methods to quantify approximate behaviours in broader classes of rational relations.

\bibliography{reference} 

\end{document}